\documentclass[runningheads]{llncs}
\usepackage{microtype,verbatim,dsfont}
\usepackage{hyperref,xspace,xcolor}
\usepackage{todonotes}
\usepackage{graphicx}
\usepackage[subrefformat=parens,labelfont=default]{subcaption}
\usepackage{enumitem}
 
\let\endproof\relax
\usepackage{amsmath,amssymb,amsthm}
\usepackage{enumitem}

\usepackage{algorithm} 
\usepackage{algpseudocode}
\usepackage{comment}
\algtext*{EndIf}
\algtext*{EndFor}

\newtheorem*{problem*}{Problem}
\newtheorem*{theorem*}{Theorem}

\newcommand{\OPT}{\textnormal{OPT}}
\newcommand{\ALG}{\textnormal{ALG}}
\newcommand{\N}{\mathbb{N}}
\newcommand{\priority}{\ell}
\newcommand{\partition}{S}
\newcommand{\sparsification}{H}

\begin{document}
\title{Multi-Priority Graph Sparsification\thanks{Supported in part by NSF grants CCF-1740858, CCF-1712119, and DMS-1839274.}
}


%
%
\author{Reyan Ahmed
\and
Keaton Hamm
\and
Stephen Kobourov
\and
Mohammad Javad Latifi Jebelli
\and
Faryad Darabi Sahneh
\and
Richard Spence
}
\institute{~}

\authorrunning{R. Ahmed et al.}
%
%
\maketitle              
\begin{abstract}
A \emph{sparsification} of a given graph $G$ is a sparser graph (typically a subgraph) which aims to approximate or preserve some property of $G$. Examples of sparsifications include but are not limited to spanning trees, Steiner trees, spanners, emulators, and distance preservers. Each vertex has the same priority in all of these problems.
However, real-world graphs typically assign different ``priorities'' or ``levels'' to different vertices, in which higher-priority vertices require higher-quality connectivity between them. Multi-priority variants of the Steiner tree problem have been studied in prior literature but this generalization is much less studied for other sparsification problems.
In this paper, we define a generalized multi-priority problem and present a rounding-up approach that can be used for a variety of graph sparsifications.
Our analysis provides a systematic way to compute approximate solutions to multi-priority variants of a wide range of graph sparsification problems given access to a single-priority subroutine.

\end{abstract}

\section{Introduction}
A \emph{sparsification} of a graph $G$ is a  graph $\sparsification$ which preserves some property of $G$. Examples of sparsifications include spanning trees, Steiner trees, spanners, emulators, distance preservers, $t$--connected subgraphs, and spectral sparsifiers. Many sparsification problems are defined with respect to a given subset of vertices $T \subseteq V$ which we call \emph{terminals}: e.g., a \emph{Steiner tree} over $(G, T)$ requires a tree in $G$ which spans $T$.  
{\color{black}Most of the corresponding optimization problems of these sparsifications
are NP-hard to compute optimally, so we often seek approximation algorithms or other algorithms which yield good solutions in practice.}

{\color{black}Different sparsifications can play a significant role in tackling real-world network design problems containing a large number of nodes and edges.}
For example, networks arising in real-world applications can be of vast scale, and often contain millions of vertices and even more edges (social networks, epidemiological networks, road networks, etc.). Visualizing such large  networks at once with all important information is impossible, hence it is desirable to have a multi-priority structure for the network, in which low priority vertices and edges capture finer detail, while higher priority vertices and edges form a significantly sparser network, which nonetheless still represents the underlying structure.

{\color{black}In this paper, we are interested in generalizations of sparsification problems where each vertex possesses one of $k+1$ different \emph{priorities} (between 0 and $k$, where $k$ is the highest), in which the goal is to construct a graph $\sparsification$ such that (i) every edge in $\sparsification$ has a \emph{rate} between 1 and $k$ inclusive, and (ii) For all $i \in \{1,\ldots,k\}$, the edges in $\sparsification$ of rate $\ge i$ constitute a given type of sparsifier over the vertices whose priority is at least $i$.} Throughout, we assume a vertex with priority 0 need not be included. This multi-priority problem has been studied in the context of Steiner or multicast trees~\cite{Balakrishnan1994,Charikar2004ToN,Chuzhoy2008}, where the objective is to connect a source to a set of heterogeneous receivers, each with a certain priority request. However, this generalization is far less studied for other types of sparsifications. 
We investigate multi-priority generalizations of other graph problems and provide approximation algorithms for these problems.

\subsection{Problem definition}\label{subsection:definition}
Here, we will consider a general range of sparsification problems. A sparsification $\sparsification$ is \textit{valid} if it satisfies a set of constraints that depends on the type of sparsification. Given $G$ and a set of terminals $T$, let $\mathcal{F}$ be the set of all valid sparsifications $\sparsification$ of $G$ over $T$. Throughout, we will assume that $\mathcal{F}$ satisfies the following \textit{general constraints} that must hold for all types of sparsification we consider in this article: for all $\sparsification \in \mathcal{F}$:
\begin{itemize}
    \item $\sparsification$ contains all terminals $T$ in the same connected component, and
    \item $\sparsification$ is a subgraph of $G$.
\end{itemize}
Besides these general constraints, there are additional constraints that depend on the specific type of sparsification as described below. Note that $|\mathcal{F}|$ is usually too large to enumerate, but in many cases $\mathcal{F}$ can be implicitly stated via constraints, and we will assume that checking if $\sparsification \in \mathcal{F}$ can be done in polynomial-time. Several such sparsification problems we focus on include: \vspace{3pt}

{\color{black}
\noindent \textit{Steiner trees:} A Steiner tree over $(G, T)$ is a subtree that spans $T$. In this case, we may let $\mathcal{F}$ be the set of all Steiner trees over $(G, T)$. The specific constraint for this problem is the sparsification $\sparsification$ which is a Steiner tree over $(G, T)$; we denote this constraint by \textit{tree constraint}.
}

\noindent \textit{Subset spanners and distance preservers:} A spanner is a subgraph which approximately preserves pairwise distances in the original graph $G$. A \emph{subset spanner} needs only approximately preserve distances between a subset $T \subseteq V$ of vertices. Two common types of spanners include \emph{multiplicative} $\alpha$-spanners, which preserve distances in $G$ up to a multiplicative $\alpha$ factor, and \emph{additive} $+\beta$ spanners, which preserve distances up to additive $+\beta$ error. A \emph{distance preserver} is a special case of the spanner where distances are preserved exactly. The specific constraints are basically the distance constraints applied from the problem definition. For example, for multiplicative $\alpha$-spanners the constraints are that the distance in $\sparsification$ between any pair of vertices is no more than $\alpha$ times the distance in $G$. {\color{black}We denote these types of constraints by \textit{distance constraints}}. \vspace{3pt}

The above problems are widely studied in literature; see surveys~\cite{spannersurvey,hauptmann2013compendium}. An example sparsification which we will not consider in the above framework is the \emph{emulator}, which approximates distances but is not necessarily a subgraph.

{\color{black}
In a standard weighted graph $G = (V, E)$, given a set of edges $E' \subseteq E$, the weight of the induced subgraph just depends on the weight of edges in $E'$. In this paper we study a problem where the weight not only depends on the edge weights but also on the rate of the edges. We denote the weight of the edge $e$ having rate $r$ by $w(e, r)$. Different strategies can be used to increase the weight of the edge as the rate increases. One natural setting is the linear increment: $w(e, r) = r \ w(e, 1)$. We can set $w(e, 1) = w(e)$, the input weight of $e$. It is also possible to consider $w(e, 1) = 1$ if the graph is unweighted. 
We assess the quality of a sparsification $\sparsification$ by $\text{weight}(H) := \sum_{e \in E(H)}$ $w(e, R(e))$.
We define a $k$-priority sparsification as follows, where $[k] := \{1,2,\ldots,k\}$.

\begin{definition}[$k$-priority sparsification] \label{def:main}
Let $G(V,E)$ be a graph, where each vertex $v \in V$ has priority $\priority(v) \in [k] \cup \{0\}$. Let $T_i := \{v \in V \mid \priority(v) \ge i\}$. 
Let $w(e)$ be the edge weight of edge $e$. The weight of an edge having rate $i$ is denoted by $w(e, i) = i \ w(e, 1) = i \ w(e)$.
For $i \in [k]$, let $\mathcal{F}_i$ denote the set of all valid sparsifications over $T_i$. A subgraph $\sparsification$ with edge rates $R:E(H) \to [k]$ is a $k$-priority sparsification if for all $i \in [k]$, the subgraph of $\sparsification$ induced by all edges of rate $\ge i$ belongs to $\mathcal{F}_i$.
We assess the quality of a sparsification $\sparsification$ by $\text{weight}(H) := \sum_{e \in E(H)}$ $w(e, R(e))$.
\end{definition}

 Note that $\sparsification$ induces a nested sequence of $k$ subgraphs, and can also be interpreted as a \emph{multi-level} graph sparsification~\cite{MLST2018}. A road map serves as a good analogy of a multi-level sparsification, as zooming out only displays highways and other major roads (Figure~\ref{FIG:Arizona}). Figure~\ref{fig:mlgs-example} shows an example of 2-priority sparsification with distance constraints where $\mathcal{F}_i$ is the set of all subset $+2$ spanners over $T_i$; that is, the vertex pairs of $T_i$ is connected by a path in $\sparsification_i$ at most 2 edges longer than the corresponding shortest path in $G$. Similarly, Figure~\ref{fig:mlst-example} shows an example of 2-priority sparsification with a tree constraint.
}

\begin{figure}[h!]
    \centering
    \includegraphics[width=0.3\textwidth]{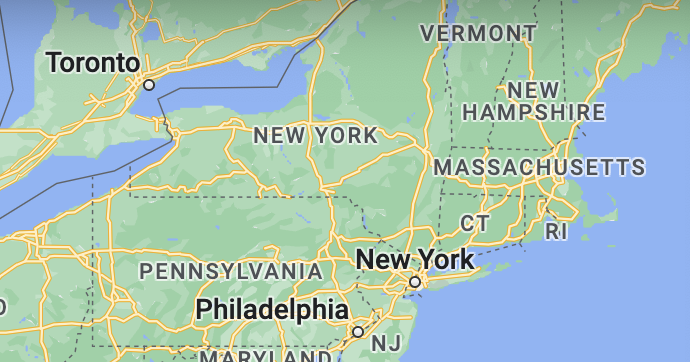}\;\;
    \includegraphics[width=0.3\textwidth]{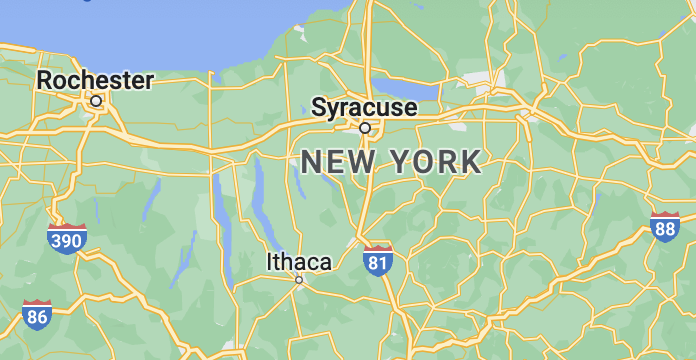}
    \;\;\includegraphics[width=0.3\textwidth]{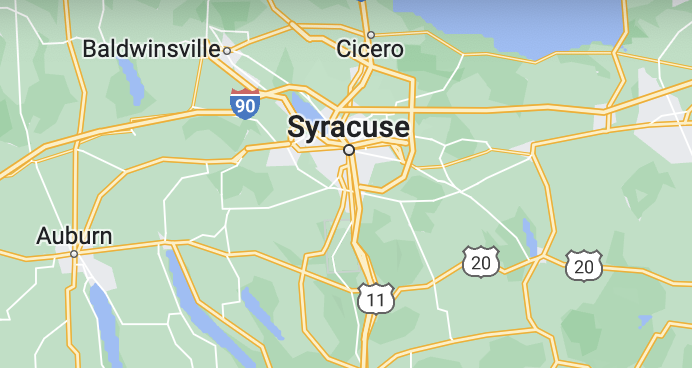}
    \caption{\color{black}Three different zoom levels of a map of New York (\emph{Map data: Google}).}
    \label{FIG:Arizona}
\end{figure}

\begin{figure}[h]
    \centering
    \includegraphics[width=0.9\textwidth]{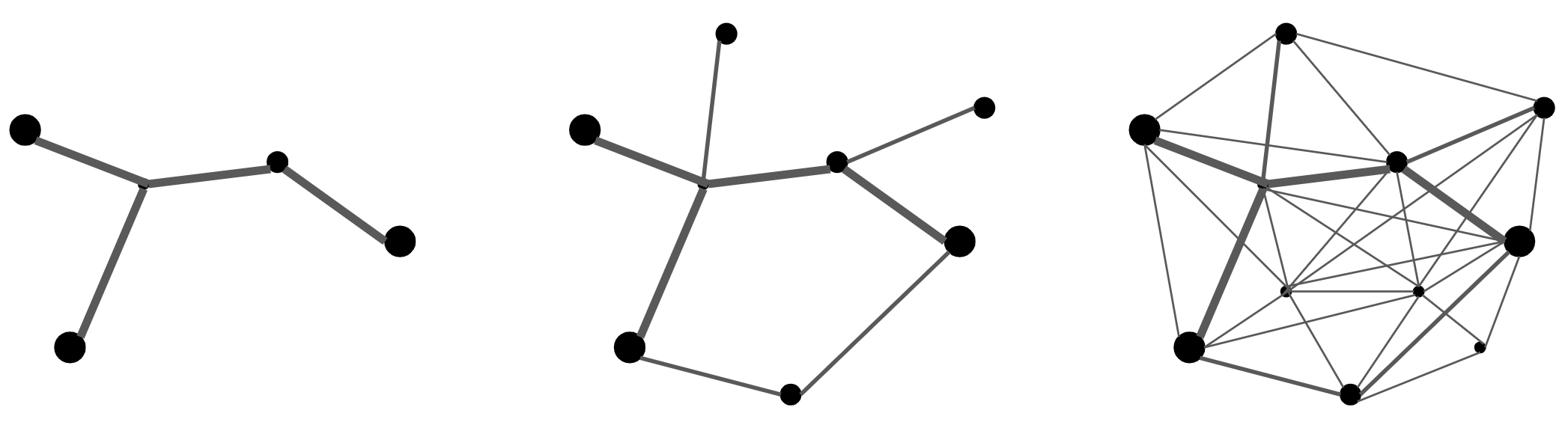}
    \caption{\color{black}\emph{Right:} A graph $G$ with $k=2$ priorities, with four and three vertices of priority 1 and 2, respectively (indicated using small and large circles). \emph{mid:} The subgraph $\sparsification_1$ with edges of rates 2 and 1 (indicated using the thickness) of a 2-priority sparsification with distance constraints. \emph{left:} The subgraph $\sparsification_2$ with edges of rate 2 of the same 2-priority sparsification.}
    \label{fig:mlgs-example}
\end{figure}

\begin{figure}[h]
    \centering
    \includegraphics[width=0.9\textwidth]{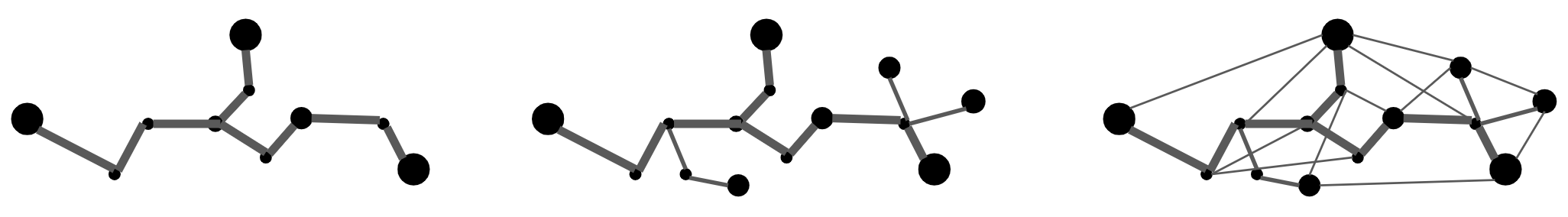}
    \caption{\color{black}\emph{Right:} A graph $G$ with $k=2$ priorities, with four and three vertices of priority 1 and 2, respectively (indicated using small and large circles). \emph{mid:} The subgraph $\sparsification_1$ with edges of rates 2 and 1 (indicated using the thickness) of a 2-priority sparsification with a tree constraint. \emph{left:} The subgraph $\sparsification_2$ with edges of rate 2 of the same 2-priority sparsification.}
    \label{fig:mlst-example}
\end{figure}

{\color{black}The linearly increasing $k$-priority instance is motivated by a natural large network visualization problem. Semantic zooming features similar to the Google map (Figure~\ref{FIG:Arizona}) are desirable while visualizing large networks. In semantic zooming, when an object appears at a particular level, it should not suddenly disappear after zooming in. In the context of network visualization, we can say that if an edge appears at a particular rate or level, then it should also appear in the lower levels. In other words, an edge can appear in multiple levels and the total weight is the highest level of appearance times the edge weight in level 1.}

Definition~\ref{def:main} is intentially open-ended to encompass a wide variety of sparsification problems. The $k$-priority problem is a generalization of many NP-hard problems, for example, Steiner trees, spanners, distance preservers, etc. These classical problems can be considered different variants of the 1-priority problem. Hence, the $k$-priority problem can not be simpler than the $1$-priority problem. In this paper, we are mainly interested in the following problem: how much harder is it to compute a $k$-priority sparsification, compared to the corresponding 1-priority sparsification problem?


We use the terms cost and weight interchangeably, as from an optimization point of view the term ``cost'' is more intuitive. Let $\OPT$ be an optimal solution to the $k$-priority problem and the cost of $\OPT$ be $\text{weight}(\OPT)$.

\begin{problem*} \label{prob:main}
Given $\langle G, P, w\rangle$ consisting of a graph $G$ with vertex priorities $\priority:V \to [k]\cup \{0\}$, can we compute a $k$-priority sparsification whose weight is small compared to $\text{weight}(\OPT)$?
\end{problem*}

\subsection{Related work}
\label{section:related}
The case where $\mathcal{F}_i$ consists of all Steiner trees over $T_i$ is known under different names including Priority Steiner Tree~\cite{Chuzhoy2008}, Multi-level Network Design~\cite{Balakrishnan1994}, Quality-of-Service Multicast Tree~\cite{Charikar2004ToN,Karpinski2005}, and Multi-level Steiner Tree~\cite{MLST2018}.

Charikar et al.~\cite{Charikar2004ToN} give two $O(1)$-approximations for the Priority Steiner Tree problem using a rounding approach which rounds the priorities of each terminal up to the nearest power of some fixed base (2 or $e$), then using a subroutine which computes an exact or approximate Steiner tree. If edge weights are arbitrary with respect to rate (not necessarily increasing linearly w.r.t. the input edge weights), the best known approximation algorithm achieves ratio $O(\min\{\log |T|, k\rho\}$~\cite{Charikar2004ToN,sahneh2021pst} where $\rho \approx 1.39$~\cite{Byrka2013} is an approximation ratio for the edge-weighted Steiner tree problem. On the other hand, the Priority Steiner tree problem cannot be approximated with ratio $c \log \log n$ unless $\text{NP} \subseteq \text{DTIME}(n^{O(\log \log \log n)})$~\cite{Chuzhoy2008}.

Ahmed et al.~\cite{MLGS_proceeding} describe an experimental study for 
the $k$-priority problem in the case where $\mathcal{F}_i$ consists of all subset multiplicative spanners over $T_i$. They show that simple heuristics for computing multi-priority spanners already perform nearly optimally on a variety of random graphs.
Multi-priority variants of additive spanners have also been studied~\cite{ahmed_et_al:LIPIcs.SEA.2021.16}, although with objective functions that are more restricted than our setting. 

\subsection{Our contribution}\label{section:our_contribution}
We have extended the rounded-up approach provided by Charikar et al.~\cite{Charikar2004ToN}. The original approach has several limitations:\begin{itemize}
    \item The algorithm only ensures that for each priority the terminals are connected. By removing the redundant paths, one can generate a tree-like structure, but it is not obvious how to generate  sparsifications for distance preservers or spanners.
    \item The algorithm computes a solution for different priorities independently. However, we may need to satisfy constraints between different priorities. For example, in graph spanners, we have to maintain the distance constraints between pairs of vertices with different priorities. The independent solutions approach does not easily handle distance constraints between different priorities.
\end{itemize}

\noindent We generalize the algorithm to deal with these limitations and the result can be applied not only Steiner trees but to other sparsifiers. Specifically:
\begin{itemize}
    \item We define a merging operation that can handle graph spanners, distance preservers, and similar structures in addition to Steiner trees.
    \item We propose different partitioning of the terminals that helps to satisfy the distance constraints between different priorities and study the trade-offs between different partitioning techniques.
    \item {\color{black}We prove an approximation guarantee for all considered sparsifications using proof by induction that is independent of the partitioning method.}
\end{itemize}


\section{A general approximation for $k$-priority sparsification}

{\color{black}In this section, we generalize the rounding approach of~\cite{Charikar2004ToN}.
The approach has two main steps: the first step rounds up the priority of all terminals to the nearest power of 2; the second step computes a solution independently for each rounded-up priority and merges all solutions from the highest priority to the lowest priority. Each of these steps can make the solution at most two times worse than the optimal solution. Hence, overall the algorithm is a $4$-approximation, we provide the pseudocode of the algorithm below. }

\begin{algorithm}
\caption*{\textbf{Algorithm} $k$-priority Approximation($G = (V, E)$)}
\begin{algorithmic}
\State // Round up the priorities
\For{each terminal $v \in V$}
\State Round up the priority of $v$ to the nearest power of 2
\EndFor
\State // Independently compute the solutions
\State Compute a partition $\partition_1, \partition_2, \cdots, \partition_k$ from the rounded-up terminals
\For{each partition $\partition_i$}
\State Compute a $1$-priority solution on partition $\partition_i$
\EndFor
\State // Merge the independent solutions
\For{$i \in \{k, k-1, \cdots, 1\}$}
\State Merge the solution of $\partition_i$ to the solutions of lower priorities
\EndFor
\end{algorithmic}
\end{algorithm}

Independently computed solutions introduce some complicated situations for the $k$-priority problem. For example, consider the most natural \textit{partitioning} of the terminals to different priorities: if $t_i$ and $t_j$ are two terminals having priority $i$ and $j$ respectively (where $i<j$), then assign $t_i$ to $\partition_i$ and $t_j$ to $\partition_j\setminus\partition_i$. Here, $\partition_i$ and $\partition_j$ are the partitions having priority $i$ and $j$, respectively. In other words, this is an \textit{exclusive} partitioning: each terminal vertex is assigned to exactly one set of terminals. In the second step, we can compute a solution for each terminal set independently. Although exclusive partitioning is a natural approach, it may generate invalid solutions. If we select a vertex with the highest priority as a root and add it to each partition, then the solution based on exclusive partitioning will satisfy the general constraints mentioned in Section~\ref{subsection:definition}. It will also satisfy the additional constraints for the Steiner tree after merging the solutions. However, it may not satisfy the additional constraints of other types of problems: for example, in graph spanners, we may need to satisfy the distance constraints between two terminal vertices that belong to two different partitions.

We now propose two partitioning techniques that will guarantee valid solutions. The first one is the \textit{inclusive} partitioning: each terminal $t_j$ is assigned to each partition in $\{\partition_i: i \leq \priority(t_j)\}$. In other words, $\partition_i = T_i$ for all $i$.

\begin{definition}
An inclusive partitioning of the terminal vertices of a $k$-priority instance assigns each terminal $t_j$ to each partition in $\{\partition_i: i \leq \priority(t_j)\}$.
\end{definition}

We propose another partitioning that is based on pairs of terminals. Consider a pair of terminals $t_i$ and $t_j$ and w.l.o.g. let $\priority(t_i) \leq \priority(t_j)$. Then we assign the priority to this pair equal to $\min(\priority(t_i), \priority(t_j))$. We partition the pairs of terminals for each priority and refer to it as the \textit{pairwise} partitioning.

\begin{definition}
A pairwise partitioning of the terminal vertices of a $k$-priority instance assigns the priority of each pair of terminals $t_i$ and $t_j$ equal to $\min(\priority(t_i),$ $ \priority(t_j))$. Based on this assignment, we can create a partitioning; $\forall k \partition_k = \{(t_i, t_j): min($ $\priority(t_i), \priority(t_j))=k\}$.
\end{definition}

We compute partitions $\partition_1, \partition_2, \cdots, \partition_k$ from the terminal sets $T_1, T_2,$ $\cdots, T_k$ and use them  to compute the independent solutions. 
{\color{black}Here, we require one more assumption: given $1 \le i < j \le k$ and two partitioned sets $\partition_i, \partition_j$, any two sparsifications of rate $i$ and $j$ can be ``merged'' to produce a third sparsification of rate $i$.} Specifically, if $\sparsification_i \in \mathcal{F}_i$, and $\sparsification_j \in \mathcal{F}_j$, then there is a graph $\sparsification_{i,j} \in \mathcal{F}_i$ such that $\sparsification_j \subseteq H_{i,j} \subseteq H_i \cup H_j$. For the above sparsification problems (e.g., Steiner tree, spanners), we can often let $\sparsification_{i,j}$ be the union of $\sparsification_i$ and $\sparsification_j$, though edges may need to be pruned to ensure that $\sparsification_{i,j}$ is a Steiner tree (by removing cycles).

\begin{definition}
Let $\partition_i$ and $\partition_j$ be two partitions where $i<j$. Let $\sparsification_i$ and $\sparsification_j$ be the independently computed solution for the terminal set $T_i$ and $T_j$ respectively. We say that the solution $\sparsification_j$ is merged with solution $\sparsification_i$ if we complete the following two steps:\begin{enumerate}
    \item If an edge $e$ is not present in $\sparsification_i$ but present in $\sparsification_j$, then we add $e$ to $\sparsification_i$.
    \item {\color{black}If there is a tree constraint, then prune some lower-rated edges to ensure there is no cycle.}
\end{enumerate}
\end{definition}

{\color{black}For a tree constraint, we need to prune lower-rated edges since pruning higher-rated edges may disconnect the tree. More specifically, for each pair of terminals $u$ and $v$ in $\partition_j$, we check if there exists more than one path in $\sparsification_j$. We prune edges until there is only one path $P$ between $u$ and $v$. Now consider the solution $\sparsification_i$. If a path between $u$ and $v$ other than $P$ exists in $\sparsification_i$, we remove more edges until $P$ is the only path between $u$ and $v$. At the end, there is exactly one path for each pair of terminals in both $\sparsification_i$ and $\sparsification_j$. We need the second step of merging particularly for sparsifications with tree constraints. Although the merging operation treats these sparsifications differently, we will later show that the pruning step does not play a significant role in the approximation guarantee.}

{\color{black}Algorithm $k$-priority Approximation computes a partition from the rounded-up terminals. The following claims show that if the algorithm computes either an inclusive or pairwise partitioning, then the algorithm provides a valid solution.}

\begin{lemma}\label{lemma:inclusive_partitioning}
{\color{black}If Algorithm $k$-priority Approximation computes an inclusive partitioning, then the algorithm provides a valid solution.}
\end{lemma}

\begin{proof}
To compute the solution of the $k$-priority instance we merge all the independent solutions, that is, if an edge is present for a particular rate $i$, then it is also present for all rates smaller than $i$. If we are computing the Steiner tree, then the merging operation ensures that we have exactly one path for each pair of terminals. Hence, we have a valid priority Steiner tree.
Now suppose that we are computing a spanner (or preserver). Consider a pair of terminals $t_i \in \partition_i$ and $t_j \in \partition_j$, w.l.o.g. we assume $\priority(t_i) \leq \priority(t_j)$. Then $t_j \in \partition_i$ since the partitioning is inclusive. Hence there is a path between $t_i$ and $t_j$ satisfying the distance constraint since we have computed an independent spanner on $\partition_i$. Hence the merged solution is valid.  
\end{proof}

\begin{lemma}\label{lemma:pairwise_partitioning}
{\color{black}If Algorithm $k$-priority Approximation computes a pairwise partitioning, then the algorithm provides a valid solution.}
\end{lemma}

\begin{proof}
After we merge the independent solutions, the general constraints will be satisfied. Also, the second step of the merging operation will make sure that for each pair of terminals there is exactly one path if we are computing priority Steiner trees. Hence, in that case, the output will be a valid priority Steiner tree. Now consider a sparsification where distance constraints must be satisfied to obtain a valid solution. Consider any pair of terminals $t_i$ and $t_j$. If $\priority(t_i) = \priority(t_j) = k$, then the partition $\partition_k$ contains this pair. Hence, when we compute the independent solution of $\partition_k$, the distance constraint of this pair is satisfied. Otherwise, let $\min(\priority(t_i), \priority(t_j))=k$. Then the distance constraint needs to be satisfied at priority $k$. Also, the pair will be in $\partition_k$. Hence, after computing the independent solution the constraint will be satisfied. Hence, in both cases, we have a valid sparsification. 
\end{proof}

Both the pairwise partitioning and inclusive partitioning are theoretically no worse than four times the optimal solution as we prove later. The proof is the same for both cases. However, in practice, pairwise partitioning will perform better than inclusive partitioning as indicated by the following claim.

\begin{lemma}\label{lemma:inclusive_pairwise}
{\color{black}The total number of distance constraints in pairwise partitioning is less than or equal to the number of distance constraints in an inclusive partitioning.}
\end{lemma}

\begin{proof}
{\color{black}In pairwise partitioning, the total number of distance constraints is equal to the total number of pairs in $\partition_1, \partition_2, \cdots, \partition_k$. On the other hand, the total number of distance constraints in inclusive partitioning is equal to $\sum_i {|S_i| \choose 2}$.}
Consider a pair of terminals $t_i$ and $t_j$ such that both $\priority(t_i)$ and $\priority(t_j)$ are not equal to 1. Then in the inclusive partitioning this pair will be considered in partitions $\partition_{\min(\priority(t_i), \priority(t_j))}$ to $\partition_1$. On the other hand, in the pairwise partitioning, this pair will be only considered in $\partition_{\min(\priority(t_i), \priority(t_j))}$. Hence, pairwise partitioning will have only one constraint for this pair and the inclusive partitioning will have more than one constraint. Overall, the pairwise partitioning will have a smaller number of constraints and better running time.
\end{proof}

{\color{black}Algorithm $k$-priority Approximation provides valid solutions for both inclusive partitioning and pairwise partitioning as shown in Lemma~\ref{lemma:inclusive_partitioning} and~\ref{lemma:pairwise_partitioning}. Lemma~\ref{lemma:inclusive_pairwise} shows that pairwise partitioning is better in terms of the number of distance constraints. We now provide an approximation guarantee for Algorithm $k$-priority Approximation that is independent of the partitioning method.}

\begin{theorem} \label{thm:rounding}
{\color{black}Consider an instance $\varphi = \langle G, P, w\rangle$ of the $k$-priority problem with linear edge weights. If we are given an oracle that can compute the minimum weight sparsification of $G$ over $T$, then with at most $\log_2 k + 1$ queries to the oracle, Algorithm $k$-priority computes a $k$-priority sparsification with weight at most $4 \ \text{weight}(\OPT)$.}
\end{theorem}
\begin{proof}
 The $k$-priority problem does not explicitly require some vertex has priority $k$, so we will assume w.l.o.g. $k$ is a power of 2. Given $\varphi$, construct the rounded-up instance $\varphi'$ which is obtained by rounding up the priority of each vertex to the nearest power of 2. Then, if $\OPT'$ is 
 an optimum solution to the rounded-up instance, we have $\text{weight}(\OPT') \le 2 \ \text{weight}(\OPT)$, since edge weights are linear.
 
 Then for each rounded-up priority $i \in \{1,2,4,8,\ldots,k\}$, compute a sparsification independently over the partition $\partition_i$, creating $\log_2 k + 1$ graphs.
 Denote these graphs $\ALG_1$, $\ALG_2$, $\ALG_4$, \ldots, $\ALG_k$. Combine these sparsifications into a single subgraph $\ALG$.
 This is done using the ``merging'' operation described earlier in this section: (i) add each edge of $\sparsification_i$ to all sparsification of lower priorities $\sparsification_{i-1}, \sparsification_{i-2}, \cdots , \sparsification_1$ and (ii) prune some edges to make sure that there is exactly one path between each pair of terminals if we are computing priority Steiner tree. 
 
It is not obvious why after this merging operation we  have a $k$-priority sparsification with cost no more than $4 \ \text{weight}(\OPT)$. The approximation algorithm computes solutions independently, which means it is unaware of the terminal sets at the lower levels. Consider the top most partition $\partition_k$ of the rounded up instance. The approximation algorithm computes an optimal solution for that partition. The optimal algorithm of the $k$-priority sparsification  computes the solution while considering all the terminals and all priorities. Let $\OPT_i$ be the minimum weighted subgraph in an optimal $k$-priority solution $\OPT$ to generate a valid sparsification on partition $\partition_i$. Then weight($\ALG_k$) $\leq$ weight($\OPT_k$), i.e., if we only consider the top partition $\partition_k$, then the approximation algorithm is no worse than the optimal algorithm. 

However, the approximation algorithm may incur additional cost when merging the edges of $\ALG_k$ in lower priorities. In the worst case, merged edges might not be needed to compute the solutions of the lower partitions (if the merged edges are used in the lower partitions in their independent solutions, then we do not need to pay any cost for the merging operation). This is because the approximation algorithm computes the solutions independently. On the other hand, in the worst case, it may happen that $\OPT_k$ includes all the edges to satisfy all the constraints of lower partitions. In this case, the cost of the optimal $k$-priority solution is $k \ $weight($\OPT_k$). If weight($\ALG_k$) $\approx$ weight($\ALG_{k-1}$) $\approx \cdots \approx$ weight($\ALG_1$) and the edges of the sparsification of a particular priority do not help in the lower priorities, then it seems like the approximation algorithm can perform around $k$ times worse than the optimal $k$-priority solution. However, the hypothesis (the edges of the sparsification of a particular priority do not help in the lower priorities) will not be true because we are considering a rounded up instance. In a rounded up instance $\partition_k = \partition_{k-1} = \cdots = \partition_{\frac{k}{2}+1}$. Hence, weight($\ALG_i$) $\leq$ weight($\OPT_i$) for $i = k, k/2, \cdots, \frac{k}{2}+1$.
 
 \begin{lemma}
 If we compute independent solutions of a rounded up $k$-priority instance and merge them, then the cost of the solution is no more than $2 \ \text{weight}(\OPT)$.
 \end{lemma}
 \begin{proof}
 Let the set of partitions be $\partition_k, \partition_{k/2}, \cdots, \partition_1$. Suppose we have computed the independent solution and merged them in lower priorities. We actually prove a stronger claim, and use that to prove the lemma. Note that in the worst case the cost of approximation algorithm is $2^k\text{weight}(\ALG_k) + 2^{k/2}\text{weight}(\ALG_{k/2}) + \cdots + \text{weight}(\ALG_1)$. And the cost of the optimal algorithm is $ \text{weight}(\OPT_k) + \text{weight}(\OPT_{k-1}) + \cdots + \text{weight}(\OPT_1)$. We show that $2^k\text{weight}(\ALG_k) + 2^{k/2}\text{weight}(\ALG_{k/2}) + \cdots + \text{weight}(\ALG_1) \leq 2 \ (\text{weight}(\OPT_k) + \text{weight}(\OPT_{k-1}) + \cdots + \text{weight}(\OPT_1))$. Let $k=2^i$. We provide a proof by induction on $i$.
 
 Base step: If $i=0$, then we have just one partition $\partition_1$. The approximation algorithm computes a sparsification for $\partition_1$ and there is nothing to merge. Since the approximation algorithm uses an optimal algorithm to compute independent solutions, weight($\ALG_1$) $\leq 2 \ $ weight($\OPT_1$).
 
 Inductive step: We assume that the claim is true for $i = j$ which is the induction hypothesis. Hence $2^j \text{weight}(\ALG_{2^j}) + 2^{j-1} \text{weight}(\ALG_{2^{j-1}}) + \cdots + \text{weight}(\ALG_1) \leq 2\ ( \text{weight}(\OPT_{2^j}) + \text{weight}(\OPT_{2^{j}-1}) + \cdots + \text{weight}(\OPT_1) )$. We now show that the claim is also true for $i = j+1$. In other words, we have to show that $2^{j+1} \text{weight}(\ALG_{2^{j+1}}) + 2^j \text{weight}(\ALG_{2^{j}}) + \cdots + \text{weight}(\ALG_1) \leq 2\ ( \text{weight}(\OPT_{2^{j+1}}) + \text{weight}(\OPT_{2^{j+1}-1}) + \cdots + \text{weight}(\OPT_1) )$. We know,
 \begin{align*}
 \text{L.H.S.} &= 2^{j+1} \text{weight}(\ALG_{2^{j+1}}) + 2^j \text{weight}(\ALG_{2^{j}}) + \cdots + \text{weight}(\ALG_1) \\
 &= \text{weight}(\ALG_{2^{j+1}}) + \text{weight}(\ALG_{2^{j+1}}) + \cdots + \text{weight}(\ALG_{2^{j+1}}) \\
 &+ 2^j\text{weight}(\ALG_{2^{j}}) + 2^{j-1}\text{weight}(\ALG_{2^{j}-1}) + \cdots + \text{weight}(\ALG_1)) \\
 &\leq \text{weight}(\OPT_{2^{j+1}}) + \text{weight}(\OPT_{2^{j+1}}) + \cdots + \text{weight}(\OPT_{2^{j+1}}) \\
 &+ 2^j\text{weight}(\ALG_{2^{j}}) + 2^{j-1}\text{weight}(\ALG_{2^{j}-1}) + \cdots + \text{weight}(\ALG_1)) \\
 &= 2^{j+1} \text{weight}(\OPT_{2^{j+1}}) + 2^j \text{weight}(\ALG_{2^{j}}) + \cdots + \text{weight}(\ALG_1) \\
 &= 2 \ ( \text{weight}(\OPT_{2^{j+1}}) + \text{weight}(\OPT_{2^{j+1}-1}) + \cdots + \text{weight}(\OPT_{2^{j}+1}) ) \\
 &+ 2^j\text{weight}(\ALG_{2^{j}}) + 2^{j-1}\text{weight}(\ALG_{2^{j}-1}) + \cdots + \text{weight}(\ALG_1)) \\
 & \leq 2 \ ( \text{weight}(\OPT_{2^{j+1}}) + \text{weight}(\OPT_{2^{j+1}-1}) + \cdots + \text{weight}(\OPT_{2^{j}+1}) ) \\
 &+ 2 \ ( \text{weight}(\OPT_{2^j}) + \text{weight}(\OPT_{2^{j}-1}) + \cdots + \text{weight}(\OPT_1) ) \\
 & = 2 \ ( \text{weight}(\OPT_{2^{j+1}}) + \text{weight}(\OPT_{2^{j+1}-1}) + \cdots + \text{weight}(\OPT_1) ) \\
 &= \text{R.H.S.}
 \end{align*}
 
 Here, the second equality is just a simplification. The third inequality uses the fact that an independent optimal solution has a cost lower than or equal to any other solution. The fourth equality is a simplification, the fifth inequality uses the fact that the input is a rounded up instance. The sixth inequality uses the induction hypothesis. The L.H.S. is greater than the cost of the approximation algorithm. The R.H.S. is smaller than $2 \ \text{weight}(\OPT)$.
 \end{proof}
 
 We have shown earlier that the solution of the rounded up instance has a cost of no more than $2 \ \text{weight}(\OPT)$. Combining that claim and the previous claim, we can show that the solution of the approximation algorithm has cost no more than $4 \ \text{weight}(\OPT)$.
\end{proof}

In most cases, computing the optimal sparsification is computationally difficult. If an oracle is instead replaced with a $\rho$-approximation, the rounding-up approach is a $4\rho$-approximation, by following the same proof as above.

\section{Subset spanners and distance preservers}\label{sec:subset_spanners}

Here we provide a bound on the size of subsetwise graph spanners, where lightness is expressed with respect to the weight of the corresponding Steiner tree.

A \emph{spanner} of a graph $G$ is a subgraph $\sparsification$ which approximates distances in $G$ up to some error. Specifically, given a (possibly edge-weighted) graph $G$ and $\alpha \ge 1$, we say that $\sparsification$ is a (multiplicative) $\alpha$-spanner if $d_H(u,v) \le \alpha \cdot d_G(u,v)$ for all $u,v \in V$, where $\alpha$ is the \emph{stretch factor} of the spanner and $d_G(u,v)$ is the graph distance between $u$ and $v$ in $G$. A \emph{subset spanner} over $T \subseteq V$ approximates distances between pairs of vertices in $T$ (e.g., $d_H(u,v) \le \alpha \cdot d_G(u,v)$ for all $u,v \in T$). For clarity, we refer to the case where $T=V$ as an \emph{all-pairs spanner}. The \emph{lightness} of an all-pairs spanner is defined as its total edge weight divided by $w(MST(G))$. A \emph{distance preserver} is a spanner with $\alpha = 1$.

Alth\"{o}fer et al.~\cite{althofer1993sparse} give a simple greedy algorithm which constructs an all-pairs $(2k-1)$-spanner $\sparsification$ of size $O(n^{1 + 1/k})$ and lightness $1 + \frac{n}{2k}$. The lightness has been subsequently improved; in particular Chechik and Wulff-Nilsen~\cite{chechik2018near} give a $(2k-1)(1+\varepsilon)$ spanner with size $O(n^{1 + 1/k})$ and lightness $O_{\varepsilon}(n^{1/k})$. Up to $\varepsilon$ dependence, these size and lightness bounds are conditionally tight assuming a girth conjecture by Erd\H{o}s~\cite{ErdosGirth}, which states that there exist graphs of girth $2k+1$ and $\Omega(n^{1+1/k})$ edges.

For subset spanners over $T \subseteq V$, the lightness is defined with respect to the minimum Steiner tree over $T$, since that is the minimum weight subgraph which connects $T$. We remark that in general graphs, the problem of finding a light multiplicative subset spanner can be reduced to that of finding a light spanner:

\begin{lemma} \label{lemma:subset}
Let $G$ be a weighted graph and let $T \subseteq V$. Then there is a poly-time constructible subset spanner with stretch $(2k-1)(1+\varepsilon)$ and lightness $O_{\varepsilon}(|T|^{1/k})$.
\end{lemma}
\begin{proof}
Let $\tilde{G}$ be the metric closure over $(G, T)$, namely the complete graph $K_{|T|}$ where each edge $uv \in E(\tilde{G})$ has weight $d_G(u,v)$. Let $\sparsification'$ be a $(2k-1)(1+\varepsilon)$-spanner of $\tilde{G}$. By replacing each edge of $\sparsification'$ with the corresponding shortest path in $G$, we obtain a subset spanner $\sparsification$ of $G$ with the same stretch and total weight. Using the spanner construction of~\cite{chechik2018near}, the total weight of $\sparsification'$ is $O_{\varepsilon}(|T|^{1/k}) w(MST(\tilde{G}))$. Using the well-known fact that the MST of $\tilde{G}$ is a 2-approximation for the minimum Steiner tree over $(G, T)$, it follows that the total weight of $\sparsification'$ is also $O_{\varepsilon}(|T|^{1/k})$ times the minimum Steiner tree over $(G, T)$.
\end{proof}
Thus, the problem of finding a subset spanner with multiplicative stretch becomes more interesting when the input graph is restricted (e.g., planar, or $\sparsification$-minor free). Klein~\cite{Klein06} showed that every \emph{planar} graph has a subset $(1+\varepsilon)$-spanner with lightness $O_{\varepsilon}(1)$. Le~\cite{le2020ptas} gave a poly-time algorithm which computes a subset $(1+\varepsilon)$-spanner with lightness $O_{\varepsilon}(\log |T|)$, where $G$ is restricted to be $\sparsification$-minor free.

On the other hand, subset spanners with additive $+\beta$ error are more interesting, as one cannot simply reduce this problem to the all-pairs spanner as in Lemma~\ref{lemma:subset}. It is known that every unweighted graph $G$ has $+2$, $+4$, and $+6$ spanners with $O(n^{3/2})$ edges~\cite{Aingworth99fast}, $\widetilde{O}(n^{7/5})$ edges~\cite{chechik2013new}, and $O(n^{4/3})$ edges~\cite{baswana2010additive,Knudsen14} respectively, and that the upper bound of $O(n^{4/3})$ edges cannot be improved even with $+n^{o(1)}$ additive error~\cite{abboud2017frac}.

\subsection{Subset distance preservers}
Unlike spanners, general graphs do not contain sparse distance preservers that preserve all distances exactly; the unweighted complete graph has no nontrivial distance preserver and thus $\Theta(n^2)$ edges are needed. Similarly, subset distance preservers over a subset $T \subseteq V$ may require $\Theta(|T|^2)$ edges. It is an open question whether there exists $c > 0$ such that any undirected, unweighted graph and subset of size $|T| = O(n^{1-c})$ has a distance preserver on $O(|T|^2)$ edges~\cite{bodwin2021linear}. Moreover, when $|T|=O(n^{2/3})$, there are graphs for which any subset distance preserver requires $\Omega(|T|n^{2/3})$ edges, which is $\omega(|T|^2)$ when $|T| = o(n^{2/3})$~\cite{bodwin2021linear}.

\begin{theorem}If the above open question is true, then every unweighted graph with $|T|=O(n^{1-c})$ and terminal priorities in $[k]$ has a priority distance preserver of size $4 \ \text{weight}(\OPT)$.
\end{theorem}

\section{Multi-Priority Approximation Algorithms}

In this section, we illustrate how the subset spanners mentioned in Section~\ref{sec:subset_spanners} can be used in Theorem~\ref{thm:rounding}, and show several corollaries of the kinds of guarantees one can obtain in this manner. In particular, we give the first weight bounds for multi-priority graph spanners. The case of Steiner trees was discussed~\cite{MLST2018}.

\subsection{Spanners}
If the input graph is planar, then we can use the algorithm  by Klein~\cite{Klein06} to compute a subset spanner for the set of priorities we get from the rounding approach. The polynomial-time algorithm in~\cite{Klein06} has constant approximation ratio, assuming constant stretch factor, yielding the following corollary.

\begin{corollary}
Given a planar graph $G$ and $\varepsilon>0$, there exists a rounding approach based algorithm to compute a multi-priority multiplicative $(1+\varepsilon)$-spanner of $G$ having $O(\varepsilon^{-4})$ approximation. The algorithm runs in $O(\frac{|T| \log |T|}{\varepsilon})$ time, where $T$ is the set of terminals.
\end{corollary}

The proof of this corollary follows from combining the guarantee of Klein~\cite{Klein06} with the bound of Theorem~\ref{thm:rounding}. Using the approximation result for subset spanners provided in Lemma~\ref{lemma:subset}, we obtain the following corollary.

\begin{corollary}
Given an undirected weighted graph $G$, $t\in\N$, $\varepsilon>0$, there exists a rounding approach based algorithm to compute a multi-priority multiplicative $(2t-1)(1+\varepsilon)$-spanner of $G$ having $O(|T|^{\frac{1}{\varepsilon}})$ approximation, where $T$ is the set of terminals. The algorithm runs in $O(|T|^{2+\frac1k+\varepsilon})$ time.
\end{corollary}

For additive spanners, there are algorithms to compute subset spanners of size $O(n|T|^\frac{2}{3})$, $\tilde{O}(n|T|^\frac{4}{7})$ and $O(n|T|^\frac{1}{2})$ for additive stretch $2$, $4$ and $6$, respectively~\cite{Abboud16,Kavitha2017}. Similarly, there is an algorithm to compute a near-additive subset $(1+\varepsilon,4)$--spanner of size $O(n\sqrt{\frac{|T|\log n}{\varepsilon}})$~\cite{Kavitha2017}. If we use these algorithms as subroutines in Lemma~\ref{lemma:subset} to compute subset spanners for different priorities, then we have the following corollaries.

\begin{corollary}\label{COR:AdditiveWithoutOracleRounding}
Given an undirected weighted graph $G$, there exist polynomial-time algorithms to compute multi-priority graph spanners with additive stretch 2, 4 and 6, of size $O(n|T|^\frac{2}{3})$, $\tilde{O}(n|T|^\frac{4}{7})$, and $O(n|T|^\frac{1}{2})$, respectively.
\end{corollary}

\begin{corollary}\label{COR:AlphaBetaWithoutOracleRounding}
Given an undirected unweighted graph $G$, there exists a polynomial-time algorithm to compute multi-priority $(1+\varepsilon,4)$--spanners of size $O(n\sqrt{\frac{|T|\log n}{\varepsilon}})$.
\end{corollary}

Several of the above results involving additive spanners have been recently generalized to weighted graphs; more specifically, there are algorithms to compute subset spanners in weighted graphs of size $O(n|T|^\frac{2}{3})$, and $O(n|T|^\frac{1}{2})$ for additive stretch $2W(\cdot, \cdot)$, and $6W(\cdot, \cdot)$, respectively~\cite{elkin2019almost,elkin2020improved,10.1007/978-3-030-86838-3_28}, where $W(u,v)$ denotes the maximum edge weight along the shortest $u$-$v$ path in $G$. Hence, we have the following corollary.

\begin{corollary}
Given an undirected weighted graph $G$, there exist polynomial-time algorithms to compute multi-priority graph spanners with additive stretch $2W(\cdot, \cdot)$, and $6W(\cdot, \cdot)$, of size $O(n|T|^\frac{2}{3})$, and $O(n|T|^\frac{1}{2})$, respectively.
\end{corollary}

\subsection{$t$--Connected Subgraphs}
Another example which fits the framework of Section~\ref{subsection:definition} is that of finding $t$--connected subgraphs~\cite{5438613,laekhanukit2011improved,nutov2012approximating}, in which (similar to the Steiner tree problem) a set $T\subseteq V$ of terminals is given, and the goal is to find the minimum-cost subgraph $\sparsification$ such that each pair of terminals is connected with at least $t$ vertex-disjoint paths in $\sparsification$. 
Nutov~\cite{5438613} presents an approximation algorithm for this problem giving approximation ratio $O(t^2 \log t)$. Laekhanukit~\cite{laekhanukit2011improved} improves the approximation guarantee to $O(t \log t)$  if $|T| \geq t^2$ and shows that the hardest instances of the problem are when $|T| \approx t$. Nutov~\cite{nutov2012approximating} studies the subset $t$--connectivity augmentation problem where given a graph $G$ and a $(t-1)$--connected subgraph $\sparsification$, we want to augment some edges to $\sparsification$ to make it $t$--connected. The objective is to minimize the size of the set of augmented edges.  If we use the algorithm of~\cite{laekhanukit2011improved} in Theorem~\ref{thm:rounding} to compute subsetwise $t$--connected subgraphs for different priorities, then we have following corollary.

\begin{corollary}\label{COR:kConnectedWithoutOracleRounding}
Given an undirected weighted graph $G$, using the algorithm of \cite{laekhanukit2011improved} as a subroutine in Theorem~\ref{thm:rounding} yields a polynomial-time algorithm which computes a multi-priority $t$--connected subgraph over the terimals with approximation ratio $O(t \log t)$ provided $|T| \geq t^2$.
\end{corollary}

\section{Conclusions and Future Work}\label{SEC:Conclusion}

We defined a class of $k$-priority sparsification problems and analyzed their difficulty relative to their corresponding single-priority problems. The proposed technique solves these problems using a subroutine for the corresponding single-priority problem. {\color{black}Assuming linearly increasing edge weights and an exact oracle for the single priority module, our algorithm yields a constant approximation to the optimal solution that is independent of the number of priorities $k$.} Naturally, a $\rho$-approximation can be used in place of the oracle, yielding a $O(\rho)$-approximation to the $k$-priority problem, which is again independent of the number of priorities $k$.  Since the $k$-priority sparsification problem relies on a single priority subroutine, we studied the single priority subsetwise problem for graph spanners and distance preservers. 
A feature in common for all the results in this paper is that solving the $k$-priority problem relies on single priority solutions (exact or approximate). A nice open problem is whether these $k$-priority problems can be solved directly, without relying on single priority solvers, by building the solution simultaneously for all priorities.

%
%
%
\bibliographystyle{splncs04}
\bibliography{references}

\begin{thebibliography}{10}
\providecommand{\url}[1]{\texttt{#1}}
\providecommand{\urlprefix}{URL }
\providecommand{\doi}[1]{https://doi.org/#1}

\bibitem{Abboud16}
Abboud, A., Bodwin, G.: Lower bound amplification theorems for graph spanners.
  In: Proceedings of the 27th ACM-SIAM Symposium on Discrete Algorithms (SODA).
  pp. 841--856 (2016)

\bibitem{abboud2017frac}
Abboud, A., Bodwin, G.: The 4/3 additive spanner exponent is tight. Journal of
  the ACM (JACM)  \textbf{64}(4), ~28 (2017)

\bibitem{MLST2018}
Ahmed, A.R., Angelini, P., Sahneh, F.D., Efrat, A., Glickenstein, D.,
  Gronemann, M., Heinsohn, N., Kobourov, S., Spence, R., Watkins, J., Wolff,
  A.: Multi-level {S}teiner trees. In: 17th International Symposium on
  Experimental Algorithms, {(SEA)}. pp. 15:1--15:14 (2018).
  \doi{10.4230/LIPIcs.SEA.2018.15},
  \url{https://doi.org/10.4230/LIPIcs.SEA.2018.15}

\bibitem{10.1007/978-3-030-86838-3_28}
Ahmed, R., Bodwin, G., Hamm, K., Kobourov, S., Spence, R.: On additive spanners
  in weighted graphs with local error. In: Graph-Theoretic Concepts in Computer
  Science. pp. 361--373. Springer International Publishing, Cham (2021)

\bibitem{spannersurvey}
Ahmed, R., Bodwin, G., Sahneh, F.D., Hamm, K., Jebelli, M.J.L., Kobourov, S.,
  Spence, R.: Graph spanners: A tutorial review. Computer Science Review
  \textbf{37},  100253 (2020)

\bibitem{ahmed_et_al:LIPIcs.SEA.2021.16}
Ahmed, R., Bodwin, G., Sahneh, F.D., Hamm, K., Kobourov, S., Spence, R.:
  {Multi-Level Weighted Additive Spanners}. In: Coudert, D., Natale, E. (eds.)
  19th International Symposium on Experimental Algorithms (SEA 2021). Leibniz
  International Proceedings in Informatics (LIPIcs), vol.~190, pp. 16:1--16:23.
  Schloss Dagstuhl -- Leibniz-Zentrum f{\"u}r Informatik, Dagstuhl, Germany
  (2021). \doi{10.4230/LIPIcs.SEA.2021.16},
  \url{https://drops.dagstuhl.de/opus/volltexte/2021/13788}

\bibitem{MLGS_proceeding}
Ahmed, R., Hamm, K., Jebelli, M.J.L., Kobourov, S., Sahneh, F.D., Spence, R.:
  Approximation algorithms and an integer program for multi-level graph
  spanners. In: Proceedings of the Special Event on Analysis of Experimental
  Algorithms (2019)

\bibitem{Aingworth99fast}
Aingworth, D., Chekuri, C., Indyk, P., Motwani, R.: Fast estimation of diameter
  and shortest paths (without matrix multiplication). SIAM J. Comput.
  \textbf{28},  1167--1181 (04 1999). \doi{10.1137/S0097539796303421}

\bibitem{althofer1993sparse}
Alth{\"o}fer, I., Das, G., Dobkin, D., Joseph, D., Soares, J.: On sparse
  spanners of weighted graphs. Discrete \& Computational Geometry
  \textbf{9}(1),  81--100 (1993)

\bibitem{Balakrishnan1994}
Balakrishnan, A., Magnanti, T.L., Mirchandani, P.: Modeling and heuristic
  worst-case performance analysis of the two-level network design problem.
  Management Sci.  \textbf{40}(7),  846--867 (1994).
  \doi{10.1287/mnsc.40.7.846}

\bibitem{baswana2010additive}
Baswana, S., Kavitha, T., Mehlhorn, K., Pettie, S.: Additive spanners and
  ($\alpha$, $\beta$)-spanners. ACM Transactions on Algorithms (TALG)
  \textbf{7}(1), ~5 (2010)

\bibitem{bodwin2021linear}
Bodwin, G.: New results on linear size distance preservers. SIAM Journal on
  Computing  \textbf{50}(2),  662--673 (2021). \doi{10.1137/19M123662X},
  \url{https://doi.org/10.1137/19M123662X}

\bibitem{Byrka2013}
Byrka, J., Grandoni, F., Rothvo{\ss}, T., Sanit{\`{a}}, L.: Steiner tree
  approximation via iterative randomized rounding. J. {ACM}  \textbf{60}(1),
  6:1--6:33 (2013). \doi{10.1145/2432622.2432628}

\bibitem{Charikar2004ToN}
Charikar, M., Naor, J., Schieber, B.: Resource optimization in {QoS} multicast
  routing of real-time multimedia. IEEE/ACM Transactions on Networking
  \textbf{12}(2),  340--348 (April 2004). \doi{10.1109/TNET.2004.826288}

\bibitem{chechik2013new}
Chechik, S.: New additive spanners. In: Proceedings of the twenty-fourth annual
  ACM-SIAM symposium on Discrete algorithms. pp. 498--512. Society for
  Industrial and Applied Mathematics (2013)

\bibitem{chechik2018near}
Chechik, S., Wulff-Nilsen, C.: Near-optimal light spanners. ACM Transactions on
  Algorithms (TALG)  \textbf{14}(3), ~33 (2018)

\bibitem{Chuzhoy2008}
Chuzhoy, J., Gupta, A., Naor, J.S., Sinha, A.: On the approximability of some
  network design problems. ACM Trans. Algorithms  \textbf{4}(2),  23:1--23:17
  (2008). \doi{10.1145/1361192.1361200}

\bibitem{elkin2019almost}
Elkin, M., Gitlitz, Y., Neiman, O.: Almost shortest paths and {PRAM} distance
  oracles in weighted graphs. arXiv preprint arXiv:1907.11422  (2019)

\bibitem{elkin2020improved}
Elkin, M., Gitlitz, Y., Neiman, O.: Improved weighted additive spanners. arXiv
  preprint arXiv:2008.09877  (2020)

\bibitem{ErdosGirth}
Erd{\H{o}}s, P.: Extremal problems in graph theory. In: Proceedings of the
  Symposium on Theory of Graphs and its Applications. p.~2936 (1963)

\bibitem{hauptmann2013compendium}
Hauptmann, M., Karpi{\'n}ski, M.: A compendium on Steiner tree problems. Inst.
  f{\"u}r Informatik (2013)

\bibitem{Karpinski2005}
Karpinski, M., Mandoiu, I.I., Olshevsky, A., Zelikovsky, A.: Improved
  approximation algorithms for the quality of service multicast tree problem.
  Algorithmica  \textbf{42}(2),  109--120 (2005).
  \doi{10.1007/s00453-004-1133-y}

\bibitem{Kavitha2017}
Kavitha, T.: New pairwise spanners. Theory of Computing Systems
  \textbf{61}(4),  1011--1036 (Nov 2017). \doi{10.1007/s00224-016-9736-7},
  \url{https://doi.org/10.1007/s00224-016-9736-7}

\bibitem{Klein06}
Klein, P.N.: A subset spanner for planar graphs, with application to subset
  tsp. In: Proceedings of the Thirty-eighth Annual ACM Symposium on Theory of
  Computing. pp. 749--756. STOC '06, ACM, New York, NY, USA (2006).
  \doi{10.1145/1132516.1132620},
  \url{http://doi.acm.org/10.1145/1132516.1132620}

\bibitem{Knudsen14}
Knudsen, M.B.T.: Additive spanners: A simple construction. In: Scandinavian
  Workshop on Algorithm Theory. pp. 277--281. Springer (2014)

\bibitem{laekhanukit2011improved}
Laekhanukit, B.: An improved approximation algorithm for minimum-cost subset
  k-connectivity. In: International Colloquium on Automata, Languages, and
  Programming. pp. 13--24. Springer (2011)

\bibitem{le2020ptas}
Le, H.: A ptas for subset tsp in minor-free graphs. In: Proceedings of the
  Thirty-First Annual. Society for Industrial and Applied Mathematics, USA
  (2020)

\bibitem{5438613}
Nutov, Z.: Approximating minimum cost connectivity problems via uncrossable
  bifamilies and spider-cover decompositions. In: IEEE 50th Annual Symposium on
  Foundations of Computer Science (FOCS 2009). IEEE Computer Society, Los
  Alamitos, CA, USA (oct 2009). \doi{10.1109/FOCS.2009.9},
  \url{https://doi.ieeecomputersociety.org/10.1109/FOCS.2009.9}

\bibitem{nutov2012approximating}
Nutov, Z.: Approximating subset k-connectivity problems. Journal of Discrete
  Algorithms  \textbf{17},  51--59 (2012)

\bibitem{sahneh2021pst}
Sahneh, F.D., Kobourov, S., Spence, R.: Approximation algorithms for the
  {priority Steiner tree} problem. 27th International Computing and
  Combinatorics Conference (COCOON)  (2021),
  \url{http://arxiv.org/abs/1811.11700}

\end{thebibliography}

\end{document}